\documentclass[a4paper,USenglish,cleveref, autoref, thm-restate,pdfa]{oasics-v2021}

\hideOASIcs 

\usepackage{amsmath}
\usepackage{hyperref}
\usepackage[utf8]{inputenc}

\usepackage{listings}
\usepackage{xspace}
\usepackage{caption}
\usepackage{pgfplots}
\usepackage{subcaption}
\usepackage{float}

\newcommand{\BO}[1]{{ O}\left(#1\right)}

\renewcommand{\epsilon}{\varepsilon}

\DeclareMathOperator*{\argmax}{arg\,max}

\newcommand{\HLink}[2]{\hyperref[#2]{#1~\ref*{#2}}}
\newcommand{\HLinkSuffix}[3]{\hyperref[#2]{#1\ref*{#2}{#3}}}

\newcommand{\tablab}[1]{\label{tab:#1}}
\newcommand{\tabref}[1]{\HLink{Table}{tab:#1}}

\newcommand{\figlab}[1]{\label{fig:#1}}
\newcommand{\figref}[1]{\HLink{Figure}{fig:#1}}

\newcommand{\thmlab}[1]{{\label{theo:#1}}}
\newcommand{\thmref}[1]{\HLink{Theorem}{theo:#1}}

\newcommand{\corlab}[1]{\label{cor:#1}}

\providecommand{\deflab}[1]{\label{def:#1}}
\newcommand{\defref}[1]{\HLink{Definition}{def:#1}}

\newcommand{\lemlab}[1]{\label{lemma:#1}}
\newcommand{\lemref}[1]{\HLink{Lemma}{lemma:#1}}%

\newcommand{\seclab}[1]{\label{section:#1}}
\newcommand{\secref}[1]{\HLink{Section}{section:#1}}%

\providecommand{\eqlab}[1]{}%
\renewcommand{\eqlab}[1]{\label{equation:#1}}
\newcommand{\Eqref}[1]{\HLinkSuffix{Equation~}{equation:#1}{}}

\newcommand{\im}{\texttt{IM}\xspace}
\newcommand{\bs}{\texttt{BS}\xspace}
\newcommand{\sbs}{\texttt{U-BS}\xspace}
\newcommand{\tbs}{\texttt{T-BS}\xspace}
\newcommand{\mds}{\texttt{MDS}\xspace}
\newcommand{\xc}{\texttt{X3C}\xspace}

\renewcommand{\S}{\mathcal{S}}
\newcommand{\load}[1]{\ell_{k,L}\left(#1\right)}
\newcommand{\lload}[2]{\ell_{#1}\left(#2\right)}
\newcommand{\loads}[1]{\mathcal{L}_{k,L}\left(#1\right)}
\newcommand{\lloads}[2]{\mathcal{L}_{#1}\left(#2\right)}
\newcommand{\spread}[1]{\sigma\left(#1\right)}
\newcommand{\sspread}[1]{\sigma^{(U)}\left(#1\right)}
\newcommand{\tspread}[1]{\sigma_{\gamma}^{(T)}\left(#1\right)}
\newcommand{\ttspread}[2]{\sigma_{#1}^{(T)}\left(#2\right)}

\title{On the Bike Spreading Problem}

\author{Elia Costa}{Dept. of Information Engineering, University of Padova, Italy}{}{}{}
\author{Francesco Silvestri\footnote{Corresponding author.}}{Dept. of Information Engineering, University of Padova, Italy}{silvestri@dei.unipd.it}{https://orcid.org/0000-0002-9077-9921}{}

\authorrunning{E. Costa and F. Silvestri}

\Copyright{Elia Costa \and Francesco Silvestri}



\supplement{The source code and input graphs used for the experiments are publicly available.}
\supplementdetails{GitHub link}{https://github.com/AlgoUniPD/BikeSpreadingProblem}

\funding{The paper was partially supported by UniPD SID18 grant, PRIN17 20174LF3T8, MIUR "Departments of Excellence".}

\acknowledgements{The authors would like to thank the Municipality of Padova for providing us the dataset with rides of the free-floating service in Padova.
We are also grateful to Pietro Rampazzo for providing us the grids of the city of Padova and useful suggestions on plotting maps.}

\nolinenumbers 

\begin{document}
\maketitle
\begin{abstract}
A free-floating bike-sharing system (FFBSS) is a dockless rental system where an individual can borrow a bike and returns it anywhere, within the service area.
To improve the rental service, available bikes should be distributed over the entire service area: a customer leaving from any position is then more likely to find a near bike and then to use the service.
Moreover, spreading bikes among the entire service area increases urban spatial equity since the benefits of FFBSS are not a prerogative of just a few zones.
For guaranteeing such distribution, the FFBSS operator can use vans to manually relocate bikes, but it incurs high economic and environmental costs.
We propose a novel approach that exploits the existing bike flows generated by customers to distribute bikes.
More specifically, by envisioning the problem as an Influence Maximization problem, we show that it is possible to position batches of bikes on a small number of zones, and then the daily use of FFBSS will efficiently spread these bikes on a large area.
We show that detecting these zones is NP-complete, but there exists a simple and efficient $1-1/e$ approximation algorithm;  our approach is then evaluated on a dataset of rides from the free-floating bike-sharing system of the city of Padova.
\end{abstract}

\section{Introduction}
A \emph{bike-sharing system} (BSS) is a service where an individual can rent a bike and return it after a short term.
Nowadays, almost all large cities have adopted a BSS as it is a sustainable transportation system that helps improving air pollution, public health, and traffic congestion \cite{ZHANG2018296}.
The first BSS dates back to the 1960s (with Witte Fietsen in Amsterdam), and we have seen an explosion of BSSs in the last decade, with now about 2000 operators currently managing more than 9.7 million bikes around the world \cite{bike2020}. 
The majority of BSSs are now connected to IT systems that allow to borrow bikes from a smartphone and to collect data on users and rides.
Similar systems exist for e-bikes, scooters, mopeds, and cars.

There are two major approaches to BSSs: station-based and free-floating systems. 
A \emph{station-based bike-sharing system} (SBBSS) represents the most common approach, where a user borrows a bike from a dock and returns it at another dock belonging to the same system; in an SBBSS, the set of origins and destinations of all rides is small and coincides with dock positions.
On the other hand, a \emph{free-floating bike-sharing system} (FFBSS) is a sharing model with no docks: each bike has an integrated lock that can be opened on demand with a smartphone app and returned just by closing the lock. As there are no more fixed docks, bikes can be positioned anywhere (hopefully, respecting traffic codes) and the origin/destination of a ride can be any position within the service area.
FFBSS is also an interesting solution for the first and last kilometer problem in multimodal transportation cities \cite{pal} since the average walking distance of a user to the closest free-floating bike is shorter than SBBSS.

The distribution of bikes in the service area is crucial
for to increase user satisfaction. 
In SBBSS, the main challenge is dealing with hotspots, specifically zones where several rides start (sources) or end (sinks), such as train stations or university campuses. Hotspots are critical: sources and sinks might suffer, respectively, from the lack of available bikes and parking slots in a deck. The operator needs to detect and manage these hotspots: bikes should be collected from sinks and repositioned on sources.
The number of hotspots is usually small and several efficient computational strategies have been investigated (see e.g. references in \cite{pal}).
In FBBSS, hotspots are still critical although 
customers do not experience the lack of empty slots in the returning docks as in SBBSS.

In addition to dealing with hotspots, 
FBBSS needs also to guarantee that the entire service area is covered by bikes so that a user leaving from any point can find a near available bike.
A study \cite{KBG19} has indeed shown that every additional meter of walking to a shared bike decreases a user's likelihood of using a bike by 0.194\% for short distances ($\leq 300 m$) and 1.307\% for long distances ($>300m$), implying that a user walking a distance $>500m$ for reaching the closest bike is highly unlikely to use the system. 
Moreover, if bikes are well distributed all over the service area, the spatial equity improves since the benefits of the service are not a prerogative of just some zones (e.g., city center) \cite{MOONEY201991}.
Assume the service area to be split into zones (e.g., quadrants) where the diameter of each zone is considered a reasonable walking distance (e.g., $\leq 500$m): then, the desired goal is that each zone has a sufficient number of bikes. 

To distribute bikes over the service area, an FFBSS operator could employ a fleet of vans to manually position bikes in each zone: however, this solution might be economically and environmentally unfeasible due to a large number of zones. 
In this paper, we provide a novel and alternative approach that distributes bikes over the service area by exploiting the existing customers' bike flows and hence reducing intervention by the FFBSS operator.
The idea is to detect a small number of zones, named \emph{seeds}, where bikes are more likely to be spread over the service area by the regular activity of customers during a given time interval.
The seeds represent the positions where the FFBSS operator can position batches of bikes, which will then be spread over the entire area by customers, without further interventions from the operator.
Formally, we modeled this approach as a variant of the Influence Maximization (\im) problem, which we name the  \emph{Bike Spreading (\bs) problem}.

To detect these seeds, we first define a weighted graph representing mobility flows: nodes are zones, and weighted edges represent the probability that a bike moves from one node to another. 
Then, we introduce a diffusion model to analyze how bikes move and a spread score to evaluate the quality of the final bike distribution. 
Finally, we detect a small subset of nodes that maximizes the spread score and position bikes on these nodes.
This formulation results in an NP-complete problem, but we show that there exists a $1-1/e$-approximation algorithm to the problem thanks to some properties (e.g., submodularity and monotonicity).

More specifically, the results provided in the paper are the following:
\begin{itemize}
\item In \secref{formalization}, we introduce the Bike Spreading problem and formalize two versions called \tbs and \sbs: the \tbs version aims at maximizing the number of zones with a minimum amount $\gamma>0$ of bikes, while the goal of \sbs is to uniformly distribute bikes in the service area.
\item In \secref{theory}, we analyze the theoretical properties of \tbs and \sbs: we show their NP-completeness and that \sbs satisfies the monotonicity and submodularity properties. 
By these properties and the result by Nemhauser et al. \cite{approx}, we get a simple greedy algorithm providing a $1-1/e$-approximation for \sbs.
\item In \secref{experiments}, we experimentally investigate the \bs problem by using data from the free-floating bike service of the city of Padova (Italy). We analyze the performance and quality of the solution of the greedy algorithm, compare the \sbs and \tbs problems, and make some empirical considerations on the \bs problem. 
\end{itemize}

\section{Preliminaries}

\subsection{Free-floating bike-sharing service} \seclab{review}
The studies related to bike-sharing systems mainly involve two topics: demand prediction and rebalancing. Demand analysis involves the understanding of user behavior and providing the most appropriate service (e.g., \cite{LIN2018258}).
Rebalancing has mainly focused on station-based systems (see e.g. \cite{li,forma,wang,pal}), while only a few works have addressed free-floating bikes.
Reiss and Bogenberger \cite{REISS2016341}  investigated the 
relocation strategy and a validation method on Munich’s FFBSS; their approach focused on finding the best zones where to relocate bikes to satisfy user demand and minimize bikes' idle time.
Pal and Zang \cite{pal} and Usamaa et al. \cite{usama} focused on finding the best route of relocation vans under costs and time constraints; in \cite{usama},   picking up faulty bikes was also included.
Finally, Caggiani et al. \cite{CAGGIANI2018159} proposed a forecast model aiming at reducing the number of times when a zone has fewer bikes than necessary.
These works focused on computing an optimal route to collect and relocate bikes, and on detecting hotspots where to relocate more bikes.
To the best of our knowledge, no previous works have studied how to spread bikes over the entire service area and exploited the mobility graph as in our work.

\subsection{The Influence Maximization problem} \label{sec:inf}
The \emph{Influence Maximization} (\im) problem \cite{kempe} is widely used in social network analysis to detect the most influential users that can efficiently spread information, like news or advertisements; \im is also used in epidemiology for analyzing how infections evolve via human interactions.
However, differently from news and infections, bikes do not replicate: we then need to redefine the \im framework to deal with a fixed amount of objects spreading over the graph.
In general, an Influence Maximization problem consists of three components:
\begin{itemize}
\item A directed and weighted graph where vertexes represent users, and edges represent the paths where information can propagate from a given node.
\item A diffusion model that describes how information moves in the graph. The model specifies the initial status (e.g., the nodes that initially contain the information) and how information distributes. The model advances in steps: in each step, a node detects which other nodes in the neighborhood will receive the information and propagates it.
\item An evaluation function $\sigma(\Upsilon)$, which receives a description $\Upsilon$ of how the information is distributed in the graph and provides a non-negative real value. Large values denote better and more desired distributions; e.g., $\sigma(\Upsilon)$ can represent the total number of nodes that have seen some information.
\end{itemize}
The goal of \im is to find an initial distribution that maximizes the evaluation function $\sigma(\Upsilon)$ after a given number of steps. 

The diffusion model is the critical part of \im problems, indeed they delineate the way nodes influence each other. The most commonly used diffusion models \cite{chen2011influence,survey} are {Independent Cascade}, {Linear Thresholds}, {Triggering}, and {Time Aware}: these model are {progressive} models, in the sense that once a node has been influenced it cannot change its status. There are also {non-progressive} diffusion models: some examples are the Susceptible-Infected-Susceptible (SIS) models \cite{sis}, widely used in epidemiology. 

The bike spreading problem can be viewed as a \im problem where the diffusion model propagates objects (i.e., bikes).
The main difference is that objects cannot be replicated and their amount is constant. 
In contrast, previous works on information or infection diffusion allow replicas: for instance, an infected person can infect many other persons and an image can be shared with friends in a social network.

\subsection{Approximating submodular functions} \seclab{submodular}
Nemhauser et al. \cite{approx} proved that a non-negative, monotone submodular function can be efficiently approximated by a simple greedy algorithm, within a factor $1-1/e\sim 0.63$. 
This result is used by Kempe et al. \cite{kempe} to obtain an approximate solution for the \im problem under both {Independent Cascade} and {Linear Threshold} models.

Consider a function $f$ mapping a set of elements in $U$ to a non-negative real value, i.e. $f:U^*\rightarrow \mathbb{R}^+$. 
Intuitively, a function $f$ is monotone if, by expanding a given input set, the value of the function does not decrease; a function $f$ is submodular if the marginal gain does not increase when adding more elements to the input set. 
\begin{definition}[Monotonicity]\deflab{mono}
Given a function $f:U^*\rightarrow \mathbb{R}^+$ where $U$ is a set of elements, function $f$ is said to be \emph{monotone} if, for every $S\subseteq U$ and $v\in U$, we have 
$f(S \cup \{v\} ) \geq f(S)$.
\end{definition}

\begin{definition}[Submodularity]\deflab{submod}
Given a function $f:U^*\rightarrow \mathbb{R}^+$ where $U$ is a set of elements, function $f$ is said to be \emph{submodular} if, for every $S,T\subseteq U $ with $S \subseteq T$ and $v\in U\setminus T$, we have 
$
f(S \cup \{v\}) - f(S) \geq f(T \cup \{v\}) - f(T).$
\end{definition}
The aforementioned result by Nemhauser et al. \cite{approx} is the following:
\begin{theorem}[\cite{approx}]\thmlab{approx}
Let $f:U^*\rightarrow \mathbb{R}^+$ be a non-negative, monotone submodular function where $U$ is a set of elements.
Let $k\geq 1$ be a given value and $S^*$ be a set of $k$ elements in $U$ maximizing the value of $f$ among all sets of size $k$.
Then, there exists a greedy algorithm that returns a set $S$ of $k$ elements such that $f(S)\geq 1-(1-1/k)^k f(S^*) \geq (1-1/e) f(S^*)$, where $e$ is the base of the natural logarithm.
\end{theorem}

The greedy algorithm is quite simple and it consists of extending a given set $S$ with the node $u^*$ that maximizes the marginal gain, i.e., $u^* = \argmax_{u \in U}\{f(S \cup\{u\}) - f(S)\}$ until we get a set of $k$ elements. 

\begin{lstlisting}[mathescape]
greedy($f,\; k$)
 $S = \emptyset$
 for $i = 0$ to $k-1$
 	$u^* = \argmax_{u \in U}\{f(S \cup\{u\}) - f(S)\} $ 
 	$S = S \cup \{u^*\} $
 return $S$
\end{lstlisting}

The algorithm takes time $\BO{k \Delta}$ where $\Delta$ is the maximum cost of finding the element of $U$ that maximizes the marginal gain.

\section{The Bike Spreading problem} \seclab{formalization}
In this section we formalize the Bike Spreading problem.
We first propose a general definition based on the Influence Maximization problem, and then we consider two special cases, named \sbs and \tbs problems, that target specific bike distributions.

The Bike Spreading (\bs) problem can be viewed as a special case of the Influence Maximization problem, where the entity that is distributed across the graph are items that cannot be replicated, differently from news and infections.\footnote{The model presented in our paper is not atomic: we allow a bike to be "split" since the value of a node should be understood as the expected number of bikes in that node.} We represent a city as a \emph{directed} graph $G=(V,E)$: $V$ is the set of nodes, where every node represents a different zone of the city; $E$ is the set of directed weighted edges capturing the probability of moving from one zone to another. We define $n=|V|$ and $m=|E|$, and let $\Gamma^{IN}(v)= \{w \in V | (w,v) \in E\}$ and $\Gamma^{OUT}(v)=\{w \in V | (v,w) \in E\}$ denote the set of nodes for which $v$ is the destination or the source, respectively.
The weight of an edge $e =(u,v) \in E$ represents the \emph{probability} $p_e$ that a bike in node $u$ moves from $u$ to $v$: intuitively, $p_e$ is the probability that an user renting a bike in zone $u$ ends her/his ride into zone $v$.
We use self-loops to represent bikes that stay in the same node (i.e., for bikes that are not used or are borrowed for a ride starting and ending in the same node). Then, for each node $u \in V$, we have:
\begin{equation}\eqlab{unity}
\sum_{v\in \Gamma^{OUT}(u)} p_{(u,v)} =1.
\end{equation}

We assume bikes to be initially positioned on $k\geq 1$ nodes in groups of $L\geq 1$ bikes, and we refer to the initial set of these $k$ nodes as \emph{$(k,L)$-seed}. 
Bikes can be damaged or stolen in a free-floating bike system, however, this does not significantly affect the total number of bikes in the system: therefore, we assume the total number of bikes in the graph to be fixed, and we let $B$ denote the total number of bikes, i.e. $B=k\cdot L$.

The diffusion model unfolds in $\tau\geq 1$ discrete steps. 
At any step $1 \leq t \leq \tau$ and for each vertex $v\in V$, we define the \emph{load} $\load{v,t,\S}\geq 0$ 
which represents the (expected) number of bikes in node $v$ after $t$ steps starting from a $(k,L)$-seed set $\S$.
We let $\load{v,0,\S}$ denote the initial loads: $\load{v,0,\S}=L$ for each $v\in \S$, and $\load{v,0,\S}=0$ otherwise.
At any time $0 < t\leq \tau$, bikes move according to edge directions and probabilities; each vertex load $\load{v,t,\S}$ is updated as follows:
\begin{equation*}
\load{v,t,\S}= \load{v,t-1,\S} + \Delta^{IN}(\S) - \Delta^{OUT}(\S)
\end{equation*} 
where the two rightmost terms are defined as
\begin{align*}
\Delta^{IN}(\S)=\hspace{-0.3em} \sum_{u \in \Gamma^{IN}(v)}{p_{(u,v)} \load{u,t-1,\S}}, \quad
\Delta^{OUT}(\S) = \hspace{-0.3em} \sum_{u \in \Gamma^{OUT}(v)} {p_{(v,u)} \load{v,t-1,\S}}.
\end{align*} 
Intuitively, the bikes in a node $v$ are partitioned among the outgoing edges according to their probabilities, and then the load of $v$ is decreased by the number of bikes leaving the node (i.e., $\Delta^{OUT}(\S)$) and increased by the number of bikes entering in the node (i.e., $\Delta^{IN}(\S)$).
As $ \sum_{u \in \Gamma^{OUT}(v)} p_{(v,u)} = 1$ due to  self-loops, the above update can be rewritten as follows:
\begin{align*}
\load{v,t,\S}= \sum_{u \in \Gamma^{IN}(v)}{p_{(u,v)} \load{u,t-1,\S}}.
\end{align*} 
We let $\loads{t,\S}=\left(\load{v_0,t,\S},\ldots, \load{v_{n-1},t,\S}\right)$ denote the loads of all nodes in $G$ after $1\leq t\leq \tau$ steps and with the $(k,L)$-seed $\S$.
We remark that $\load{v,t,\S}$ represents the expected number of bikes in node $v$ after $t$ steps, where all bikes in a node are spread among the neighbors nodes in every time step uniformly at random according to edge probabilities.

To measure the quality of bike distribution among nodes, we introduce the notion of spread. The \emph{spread} $\spread{\loads{t,\S}}$ after $t$ steps and with $(k,L)$-seed $\S$ is a function 
$\mathbb{R}^{n}\rightarrow \mathbb{R}$ that evaluates the quality of loads $\loads{t,\S}$. 
The actual formulation of the spread function depends on the desired distribution in a free-floating bike system.
We will see two examples of spread aiming at maximizing the number of nodes with a given minimum load, and at uniformly distributing bikes among all nodes of the graph.
We are now ready to define the {Bike Spreading problem} as:
\begin{definition}
Given a directed graph $G = (V,E)$, positive integers $L\geq 1$, $k\geq 1$ and $\tau\geq 1$, and spread $\spread{\cdot}$, 
the \emph{Bike Spreading (\bs) problem} asks for the $(k,L)$-seed $\S^*$ that maximizes $\spread{\loads{\tau,\S}}$.
More specifically, the mathematical formulation of the problem is $$
\S^*=\argmax_{\S\subseteq V, |\S| = k} \spread{\loads{\tau,\S}} $$
such that:
\begin{itemize}
\item $\displaystyle \load{v,0,\S} = L, \quad \forall v\in \S$
\item $\displaystyle \load{v,0,\S} = 0, \quad \forall v \in V\setminus \S$
\item $\displaystyle \load{v,t,\S} = \sum_{u \in \Gamma^{IN}(v)}{p_{(u,v)} \load{v,t-1,\S}}$, $\quad\forall v\in V$ and $t=1,\ldots,\tau $
\end{itemize} 
\end{definition}

We remark that in this paper we focus on small values of $\tau$ (i.e., the maximum number of steps) and thus we do not analyze the convergence of the spread for $\tau\rightarrow +\infty$. As the distribution should happen within a couple of hours from the positioning of bike batches, we expect each bike to be used just a few times and thus $\tau \leq 5$ from a practical point of view.

We now describe the two spread functions used in the paper, namely \tbs and \sbs.

\paragraph*{\tbs version}
The first example of spread leverages on the idea that a zone is considered well served by the free-floating bike system if there are at least $\gamma$ bikes (in expectation), where $\gamma>0$ is a given threshold value.
Therefore, the spread counts the number of nodes (i.e., zones) that have at least $\gamma$ bikes:
\begin{align}
\tspread{\loads{\tau,\S}} =\left|\{v \in V | \load{v,\tau,\S} \geq \gamma\} \right| \eqlab{obj0} 
\end{align}
We refer to the \bs formulation with the spread in \Eqref{obj0} as \emph{Threshold BS (\tbs)} problem.

\paragraph*{\sbs version}
The \tbs problem is an all-or-nothing approach, where only nodes receiving a sufficiently large number of bikes are relevant.
However such a solution would not penalize skewed distributions: for instance, \tbs gives the same score to a distribution of $n$ nodes with load $2\gamma$ and a skewed distribution of $n-1$ nodes with load $\gamma$ and one node with load $(n+1) \gamma$. 
Therefore, if the goal is to maximize the area covered by the service, the ideal distribution should be the uniform distribution with $B/n$ bikes per node.
We thus introduce the following spread:
\begin{equation}
\sspread{\loads{\tau,\S}} = \sum_{v\in V} \sqrt{\load{v,\tau}} \eqlab{obj1} 
\end{equation}
The maximum value of $\sspread{\loads{\tau,\S}}$ is reached when bikes are equally distributed over the graph, that is $\load{v,\tau}=B/n$.
We refer to the \bs formulation with the spread in \Eqref{obj1} as \emph{Uniform BS (\sbs)} problem.

\section{Theoretical analysis}\seclab{theory}
In this section, we first prove the NP-completeness of \tbs and \sbs problems.
Then we show that \sbs satisfies the monotonicity and submodularity properties and hence, by the result of Nemhauser et al. \cite{approx}, there exists a greedy algorithm providing a $(1-1/e)$-approximate solution to \sbs. 

\subsection{NP-completeness}
\paragraph*{\tbs is NP-complete.}
We prove the NP-completeness of \tbs with a reduction from the \textit{Minimum Dominating Set} (\textit{\mds}) problem which is known to be NP-complete even in graphs with constant vertex degree $d \geq 3$ \cite{3DS,APX_dominating}. 
Let $G = (V, E)$ be a simple, connected, undirected graph with degree $d \geq 3$ for all nodes in $V$.
A subset $S \subseteq V$ is a \emph{dominating set} if for every vertex $v \in V \setminus S$, there exists a vertex $u \in S$ such that $(v,u) \in E$; that is, every vertex outside $S$ has at least one neighbor in $S$. 
A minimum domination set of $G$ is a domination set of the smallest possible size, and we refer to its size with \emph{domination number $\gamma(G)$}.
Since the degree of every node in $G$ equals $d$, we have $\gamma(G) \leq n-d+1$, because in a set $S \subseteq V$, $|S| \geq n-d+1$, every node has at least one neighbor in $S$. 

The decision problem associated with \mds is defined as:
\begin{itemize}
\item INSTANCE: a undirected graph $G=(V,E)$ with fixed degree $d\geq 3$ and an integer $k$ with $0 < k < n-d+1$.
\item QUESTION: Does a dominating set with $\gamma(G)\leq k$ exist?
\end{itemize}

The \tbs problem can be represented by the following decision problem:
\begin{itemize}
\item INSTANCE: a directed and weighted graph $G=(V,E)$ that satisfies \Eqref{unity}, integers $k$, $L$, $\gamma$, $\tau$ and $\lambda$ with $0 < k < n$, $L\geq 1$, $\gamma>0$, $\tau\geq 1$, and $0< \lambda \leq n$.
\item QUESTION: Does a $(k,L)$-seed $\S \subseteq V$ exist in $G$ such that $\tspread{\loads{\tau,\S}} \geq \lambda$?
\end{itemize}

The reduction of \mds to \tbs is the following. 
Given an instance of the \mds problem on a graph $G$ with degree $d\geq 3$ and integer $k$, we construct a new weighted and directed graph $G'$ for the \tbs problem by directing all edges of $G$ in both directions and by adding a self-loop to each vertex, obtaining a new graph $G'=(V,E')$ with outdegree $d' = d+1$. The probability of each edge is then set to $p_e = {1}/{d'}$.
We then solve the \tbs problem on $G'$ with 
a seed with size $k$ and $L=d'$ bikes in each vertex in the seed, threshold $\gamma=1$, step number $\tau = 1$, and $\lambda = n$.
The reduction returns yes to the $d$-\mds problem if and only if \tbs returns yes, that is if there exists a $(k,d')$-seed set $\S$ in $G'$ with  $\ttspread{1}{\lloads{k,d'}{1,\S}} \geq n$. 

\begin{theorem}\thmlab{tbsNP}
The \tbs problem is \textit{NP-complete}.
\end{theorem}
\begin{proof}
The \tbs problem is in NP since a solution can be  verified in polynomial time as follows.
Let $A$ be the adjacency matrix of a graph $G$, then
matrix $A^\tau$ represents the percentage of bikes in a node $u$ that  reach a node $v$, for each $u,v\in V$, using paths of length $\tau$ (including possibly self-loops). Matrix $A^\tau$  can be computed in $\BO{n^3 \log \tau}$ time with the doubling trick.
Given a seed $\S$ and an $n$-dimensional vector $\ell_{\S}$ encoding the initial loads, then the loads at step $\tau$ can be computed with $A^\tau \cdot \ell_{\S}$ in $\BO{n^2}$ time.
Therefore $\tspread{\loads{\tau,\S}}$ for a given seed can be computed in $\BO{n^3 \log \tau}$ time.

We now show the correctness of the reduction. 
We first prove that, if there exists a dominating set $\S$ of size $k'\leq k$ in $G$, then there exits a $(k,d')$-seed in $G'$ giving $\ttspread{1}{\lloads{k,d'}{1,\S}} \geq n$. 
Let us assume for simplicity that $\S$ has size $k'=k$: it suffices to add nodes to $\S$ till reaching size $k$.
Since there are $L = d'$ bikes in each node of the seed set, $p_e = 1/d'$ and there are exactly $d'$ outgoing edges from each node, we have that each outgoing edge of a node in $\S$ is crossed by one bike. 
By definition, the dominating set $\S$ covers each vertex in $V\setminus \S$ with at least one edge; moreover, each node in $\S$ has a self-loop.
Therefore, each node in $V$ receives at least on bike and hence $\ttspread{1}{\lloads{k,d'}{1,\S}} = n$.

Conversely, if a dominating set $\S$ of size at most $k$ does not exist in $G$, then there cannot be a set $\S'$ of size $k$ in $G'$ with $\ttspread{1}{\lloads{k,d'}{1,\S}} \geq n$.
Indeed, since a dominator set of size $\leq k$ does not exists, it means that for any set $\S'$ of $\leq k$ nodes from $V$ there is at least one node $v \in V$ which is not adjacent to nodes in $\S'$. 
Therefore, for any seed set $\S'$ of $k$ nodes there exists a node $v$ not receiving any bike; 
it follows that $\lload{k,d'}{v,1,\S'}=0$ and thus $\ttspread{1}{\lloads{k,d'}{1,\S}} \leq n-1$.
\end{proof}

\paragraph*{\sbs is NP-complete.}
We use a reduction from the \textit{Exact Cover by 3-Sets} (\emph{\xc}) problem, as \mds does not work for \sbs due to its spread function.
\xc is NP-complete \cite{exact_cover}, and consists of a covering problem with sets of three elements.
Given a set $X$ with $|X| = 3q$, for some integer $q\geq 1$, and a collection $C$ of $3$-element subsets of $X$, \xc requires to decide if there exists a subset $C'\subseteq C$ such that $C'$ covers $X$ and every element of $X$ occurs in exactly one set in $C'$ (i.e., $C'$ is an \emph{exact cover} of $X$).
For clarity, consider the following example: let $X=\{1,2,3,4,5,6\}$ and $C = \{\{1,2,3\}, \{2,3,4\},\{1,2,5\}, \{2,5,6\}, \{1,5,6\}\}$;
then, the collection $C' = \{\{2,3,4\},\{1,5,6\}\} \subset C$ is an exact cover because each element in $X$ appears exactly once.
Note that, if $C=\{\{1,2,3\},\{2,4,5\},\{2,5,6\}\}$, then any collection $C'$ cannot be an exact cover: indeed, every pair of sets in $C$ shares at least one entry, and hence every collection $C'$ covers at least one element twice or more.
Note that if we do have an exact cover, $C'$ will contain exactly $q$ elements.

The decision problem associated with \xc is:
\begin{itemize}
\item INSTANCE: a set $X$, with $|X| = 3q$ for some integer $q\geq 1$, a collection $C$ of $3$-element subsets of $X$.
\item QUESTION: does a set $C' \subseteq C$ exist such that every element of $X$ occurs in exactly one member of $C'$?
\end{itemize}

The decision problem for \sbs is:
\begin{itemize}
\item INSTANCE: a directed and weighted graph $G=(V,E)$ that satisfies \Eqref{unity}, integers $k$, $L$, $\tau$, $\lambda$ with $0 < k < n$, $L\geq 1$, $\tau\geq 1$, $\lambda\geq 0$.
\item QUESTION: Does a $(k,L)$-seed $\S \subseteq V$ exist in $G$ such that $\sspread{\loads{\tau,\S}} \geq \lambda$?
\end{itemize}

The reduction from \xc to \sbs is the following.
Given an instance of the \xc, defined by a set of $3q$ elements $X=\{x_1, \ldots , x_{3q}\}$ and a collection $C=\{c_1, .., c_r\}$ of $3$-element subsets of $X$, we build a directed and weighted graph $G=(V,E)$ for the \sbs problem as follows. 
With a slight abuse of notation, we let $V=C\cup X$, 
that is, each element in $X$ and each set in $C$ are represented by a node in $V$.
Then we set $E=E_1 \cup E_2$: $E_1$ contains an edge $(c_i,x_j)$ if $x_j\in c_i$, for each $c_i\in C$ and $x_j\in X$; 
$E_2$ contains a self-loop for each node $x_i\in X$.
We observe that nodes in $X$ have only one outgoing edge, while nodes in $C$ have three outgoing edges: 
then, we set the probability of each self-loop in $E_2$ to $1$, while we set $p_e = {1}/{3}$ for all the remaining  edges in $E_1$.
We then run the \sbs problem on $G$ with $k=q$ nodes in the seed, $L=3$ bikes per node in the seed, $\tau=1$ steps, and $\lambda=3q$, and we answer yes to \xc 
if and only if \sbs returns yes, that is if there exists a $(q,3)$-seed set $\S$ in $G$ with  $\sspread{\lloads{q,3}{1,\S}} \geq 3q$.

\begin{theorem}\thmlab{sbsNP}
The \sbs problem is NP-complete.
\end{theorem}
\begin{proof}
\sbs is in NP as a solution can be checked in polynomial time $\BO{n^3\log \tau}$ as shown in the proof of \thmref{tbsNP}.

We now prove the correctness of the reduction. We first observe that any exact cover for \xc must contain $q$ entries from $C$, otherwise $X$ is not covered or an element in $X$ is covered by more than one set in $C'$.
We now prove that, if the \xc problem contains an exact cover $C'$, then there exists a $(q,3)$-seed $\S$ in $G$ with $\sspread{\lloads{q,3}{1,\S}} \geq 3q$.
Let  $\S$ be a seed set given by the nodes $c_i$  representing sets in $C'$. As each node $x\in X$ is covered by exactly one node in $C'$ and the outgoing degree of a node in $C'$ is 3, we have that $x$ receives one bike after the first step. Then:
$
\sspread{\lloads{q,3}{1,\S}} = \sum_{v\in V} \sqrt{\lload{q,3}{v,1,\S}} 
= \sum_{v \in C} 0+ \sum_{v \in X} 1 = 3q.
$

Assume now that \xc does not have an exact cover: then any $C'$ of $q$ sets from $C$ does not cover at least one point of $X$ and covers at least one point of $X$ more than once. 
Assume by contradiction that \sbs returns a seed $\S$ of $k=q$ nodes with $\sspread{\lloads{q,3}{1,\S}} \geq 3q$.
We claim that $\S$ has no nodes in $X$. If $\S$ has a node  $x \in X$, then the load of $x$ after one step is $h=\lload{q,3}{x,1,\S}=\sqrt{3}$ since every node in $X$ has only the self-loop as outgoing edge. Since  bikes positioned in nodes of $C'$ move in nodes of $X$ after one step, we get:
$$
\sspread{\lloads{q,3}{1,\S}} 
= \sqrt{h} + \sum_{v \in X\setminus \{x\}} \lload{q,3}{v,1,\S}
.$$
By the concavity of the square root, the right summation is maximized when all loads are equal  and, since $|X|=3q$, 
we get
$\sspread{\lloads{q,3}{1,\S}}\leq \sqrt{h} + \sqrt{(3q-1)(3q-h)}$. 
The right term of the inequality decreases for $h>1$, and then
$\sspread{\lloads{q,3}{1,\S}}<3q$, which is a contradiction.
Therefore, we must have that $\S$ contains only nodes in $C$.
However, since there is no exact cover, at least one node in $X$ must receive two or more bikes: by mimic the previous argument, we get that $\sspread{\lloads{q,3}{1,\S}} < 3q$. Therefore there is no $(q,3)$-seed 
giving $\sspread{\lloads{q,3}{1,\S}} \geq 3q$.
\end{proof}

\subsection{Approximation algorithms}
By the previous hardness results, we do not expect polynomial-time exact algorithms for the \tbs and \sbs problems.
In this section, by showing that \sbs satisfies the monotonicity and submodularity properties, we get that the greedy solution in \secref{submodular} gives a $(1-1/e)$-approximation algorithm for \sbs.
The \tbs version does not satisfy the submodularity property and thus similar theoretical guarantees cannot be provided: however, in the following section, we show that the greedy algorithm  experimentally provides a good approximation even for \tbs.

We start with a technical lemma and then show that  \sbs satisfies the monotonicity and submodularity properties.
 
\begin{lemma}\lemlab{sub1}
Let $S$ and $T$ be two seed sets with $S \subseteq T \subseteq V$, then $\load{v,t,T} \geq \load{v,t,S}$ for each $v \in V$ and $t\geq 0$. \end{lemma}
\begin{proof}
The proof follows by induction over the number of steps $t$. It is true in the base case when $\tau = 0$ since:
\begin{enumerate} 
\item $\load{v,0,T}= \load{v,0,S} = 0$ for each $v \notin T$; 
\item $\load{v,0,T} = \load{v,0,S} = L$ for each $v \in S$;
\item $\load{v,0,T} = L$ and $\load{v,0,S} = 0$ for each $v \in T\setminus S$.
\end{enumerate}
Now suppose that the claim is true for $t\geq 0$. 
Then at step $t+1$, we have:
\begin{align*}
\load{v,t+1,T}=&
\sum_{w \in \Gamma^{IN}(v)}{p_{(w,v)} \load{w,t,T} }\\ \geq& \sum_{w \in \Gamma^{IN}(v)}{p_{(w,v)} \load{w,t,S}} = \load{v,t+1,S}.
\end{align*} 
The claim follows. 
\end{proof}

\begin{lemma}\lemlab{sbsproperty}
Given the \sbs problem with a seed set size $k$ and given parameters $\tau$, $L$ independent of $k$, then $\sspread{\loads{\tau,\S}}$ is monotone and submodular.
\end{lemma}
\begin{proof}
The monotonicity follows from \lemref{sub1} and the monotonicity of square root. We now consider submodularity. Since the parameters $\tau$, $L$ are given and are independent of $k$, we define $\spread{\S}=\sspread{\loads{\tau,\S}}$ for notational simplicity.
By \defref{submod}, we have to prove that $
\spread{\S \cup \{v\}} - \spread{\S} \geq \spread{T \cup \{v\}} - \spread{T}$ for all $v \in V$ and $\S \subseteq T \subseteq V$.
Let $\Upsilon_{v}$ be the set of nodes in $V$ that can be reached from $v$ with a path of length $\tau$ (possibly with self-loops): nodes in $\Upsilon_{v}$ are all and only the nodes whose load can be affected by the seed in $v$.
We have:
\begin{align*}
\spread{\S \cup \{v\}} - \spread{\S} &= \sum_{u \in \Upsilon_{v}} \sqrt{\load{u,\tau,\S \cup \{v\}}} - \sqrt{\load{u,\tau,\S }}\\ 
&= \sum_{u \in \Upsilon_{v}} \sqrt{\load{u,\tau,\S}+\load{u,\tau,\{v\}}} - \sqrt{\load{u,\tau,\S }}.
\end{align*}
For any $\beta\geq 0$, we have that 
$\sqrt{\alpha + \beta} -\sqrt{\alpha}$ is a non increasing function of $\alpha$, and hence the last term of the previous inequality is lower bounded by 
$$
\sum_{u \in \Upsilon_{v}}  \sqrt{\load{u,\tau,T}+\load{u,\tau,\{v\}}} - \sqrt{\load{u,\tau,T }} = \spread{T \cup \{v\}} - \spread{T}
.$$ 
We then get 
$\spread{\S \cup \{v\}} - \spread{\S} \geq \spread{T \cup \{v\}} - \spread{T}$ that proves the submodularity of \sbs.
\end{proof}

Then, from \thmref{approx} and the above lemma, we get the following result.
\begin{corollary}\corlab{approxgreedy}
There exists a $(1-1/e)$-approximation algorithm for the \sbs problem requiring $\BO{n^3(\log \tau + k)}$ time.
\end{corollary}
\begin{proof}
The result automatically follows from \thmref{approx} by the monotonicity and submodularity properties of the spread function of  \sbs.
The greedy algorithm in \secref{submodular} gives the $(1-1/e)$-approximation.
Computing the spread for a given seed set $\S$ requires $\BO{n^2}$ time: loads can indeed be computed by the multiplication $A^\tau \cdot \ell_{\S}$, where $A$ is the adjacency matrix of graph $G$ and $\ell_{\S}$ is an $n$-dimensional  vector encoding nodes in $\S$. 
Matrix $A^\tau$ is computed in $\BO{n^3\log \tau}$ time with the doubling trick. 
Since the greedy algorithm has $k$ iterations, and each iteration checks $\BO{n}$ seeds, the claim follows.
\end{proof}

\section{Experiments}\seclab{experiments}
In \secref{buildgraph}, we explain how the input mobility graphs have been obtained, and then in \secref{expeval} we show the findings of our analysis. 
The code and the input graphs are available at \url{https://github.com/AlgoUniPD/BikeSpreadingProblem}.

\subsection{Building the mobility graphs}\seclab{buildgraph}
The input graphs used in the experiments have been built from a dataset containing all rides of the free-floating bike system in Padova (Italy) of the operator Movi by Mobike.
The dataset contains 327K rides from May 1st, 2019 to January 30th, 2020. 
Each ride is described by the anonymized user and bike ids, and by the positions and time stamps of the pick-up and drop-off points. 
The graphs were constructed by following these three steps.

\begin{enumerate}
\item Vertexes are created by snapping each pick-up and drop-off point on a 2-dimensional grid. The grid consists of cells of size $s\times s$ with $s\in\{100\; m, 500\; m\}$. The two grids create two vertex sets of size $9975$ and $399$ respectively, covering a total area of $99.75\; km^2$ (see \figref{area}).
Grid size should be understood as the maximum distance a user is willing to walk to find a bike.

\item The edge probabilities are constructed by two sets of rides: from 6:30 to 9:00 (morning rides) and from 16:00 to 20:00 (evening rides). In both sets, rides refer to all weekdays of October 2019.
For each $u,v\in V$, the probability of an edge $(u,v)$ is set to $n_{u,v}/n_u$, where $n_u$ is the total number of rides originated in the cell represented by node $u$ and $n_{u,v}$ is the total number of rides from node $u$ to node $v$.

\item The graph is pruned by removing all edges with probabilities lower than a given threshold $\eta\in\{0, 0.01, 0.1\}$. For each node $u$, the probability mass of the removed edges originating from $u$ is added to the self-loop $(u,u)$ to guarantee that the weights of outgoing edges sum to one.
Nodes with no edges or with only a self-loop are removed from the graph.
\end{enumerate}

\begin{figure}
\captionsetup[subfigure]{justification=centering}
\centering
\begin{subfigure}[a]{0.4\textwidth}
 \includegraphics[width=\textwidth]{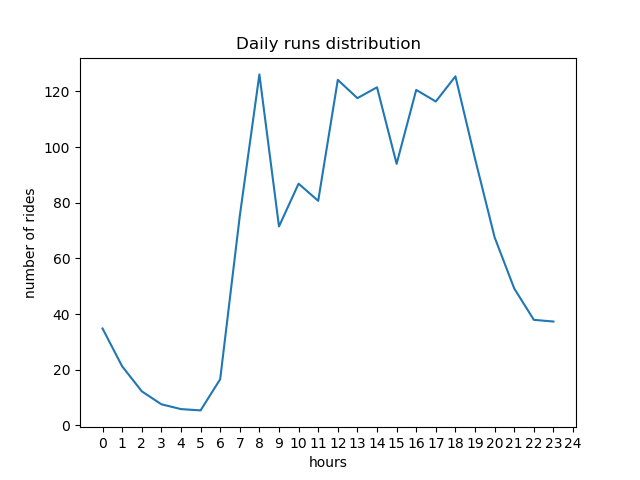}
\subcaption{}
\figlab{daily_distribution}
\end{subfigure}
\hfill
\begin{subfigure}{0.55\textwidth}
\centering
 \includegraphics[width=0.45\textwidth]{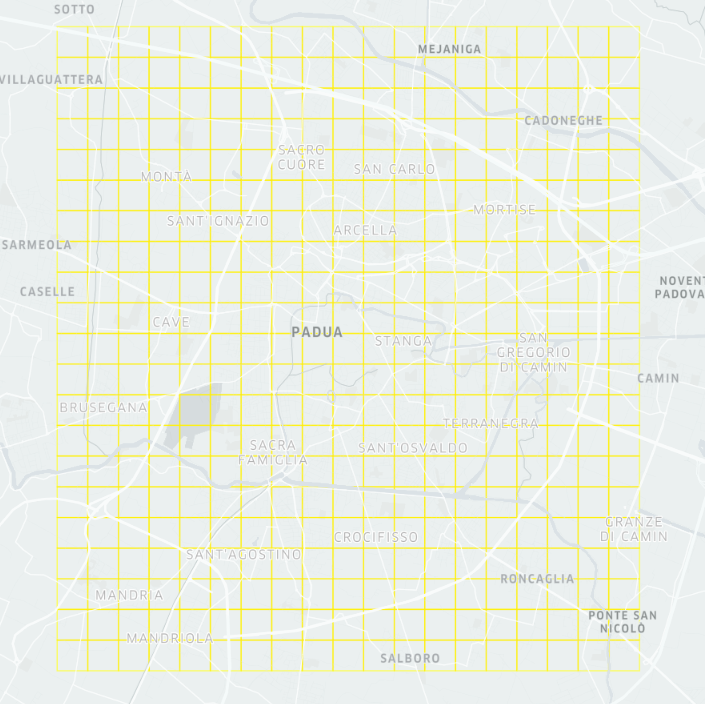}
 \includegraphics[width=0.45\textwidth]{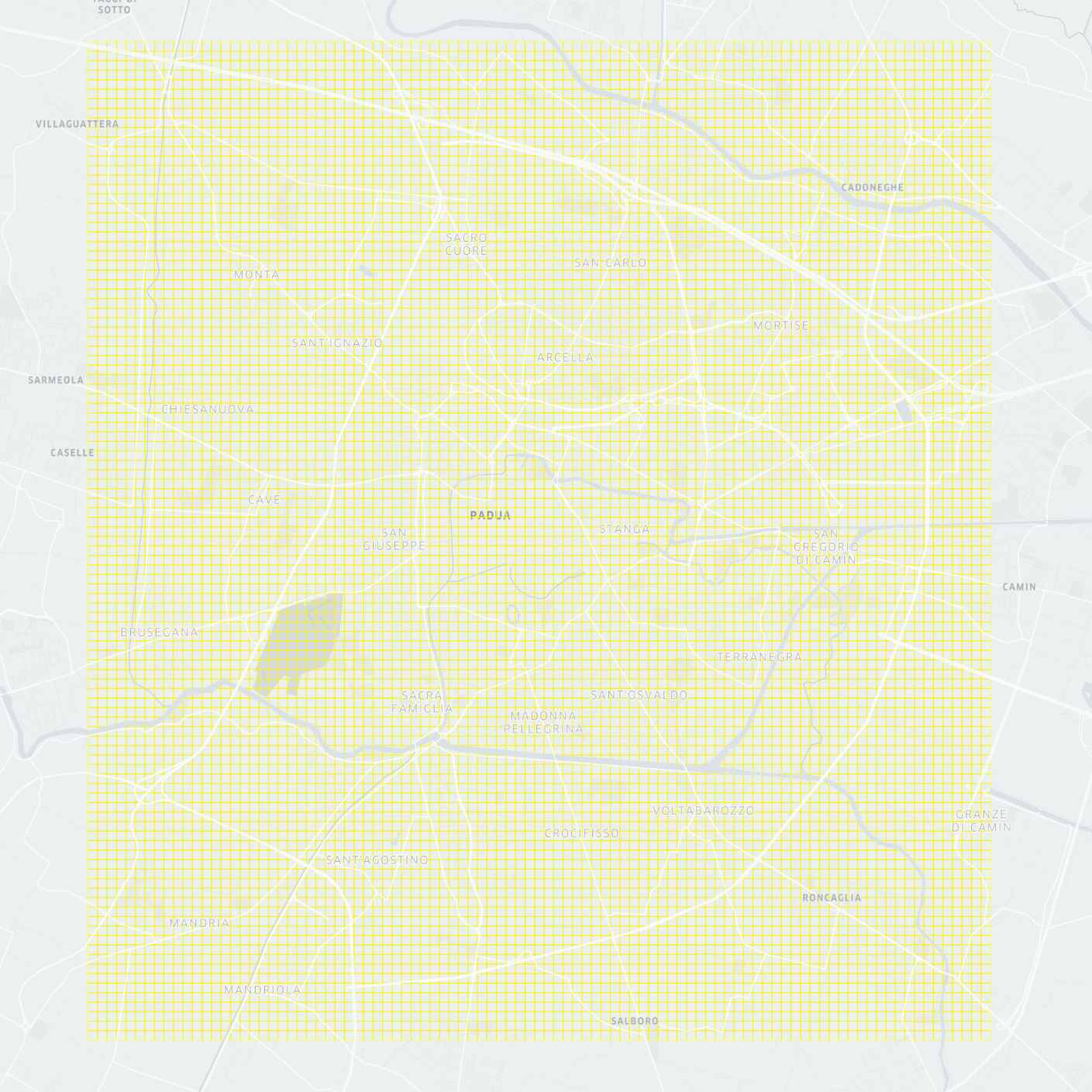}\\
 \includegraphics[width=0.45\textwidth]{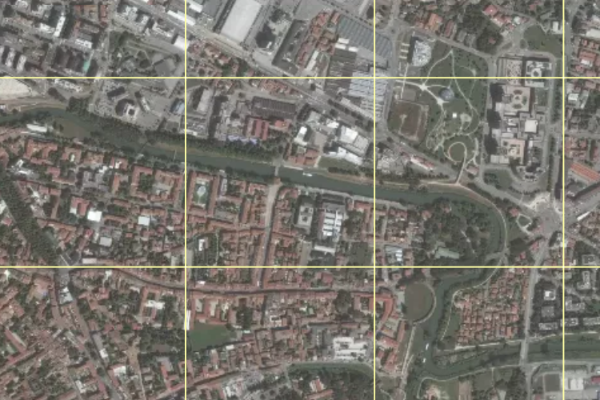}
 \includegraphics[width=0.45\textwidth]{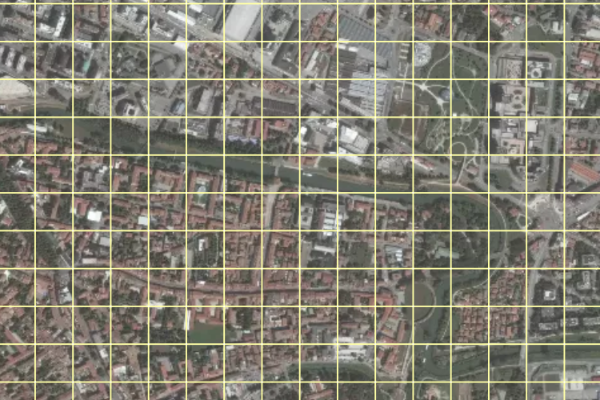}
\subcaption{}
\figlab{area}
\end{subfigure}
\caption{(a) Daily runs distribution during October 2019. There are three picks: around 8.00 (commuting home to work), 13.00 (lunch time, end of school), around 18 (commuting work to home).
(b) Padova subdivided with a grid of 500 m and 100 m, and a detail of the subdivisions. (The underlying street and satellite maps were provided by Kepler.gl.)}
\end{figure}

We thus ended up with 12 graphs $G_{s,\eta,M}$ and $G_{s,\eta,E}$ ($M$ and $E$ mean morning and evening, respectively), whose property is provided in \tabref{graphdata}. 
The intuition behind the three parameters (grid size $s$, morning or evening rides, pruning factor $\eta$) used for creating the graphs is the following.
The different grid size $s$ and pruning factor $\eta$ are used for generating graphs with a different number of nodes and edges, for testing scalability.
The graphs built using morning or evening rides give insight into the mobility flow during the day: the two 
slots are the morning and evening rush hours where workers commute to or from workplaces (see the daily rides distribution in \figref{daily_distribution}).
From a practical point of view, the weights computed from morning rides should be used if bikes are positioned in the seed nodes during the night to exploit the morning bike flow (equivalently, the evening rides should be used for bikes positioned in the afternoon).

\begin{table}
\centering
\begin{minipage}{0.47\textwidth}
\begin{tabular}{|l| r| r| r| r| r|}
\hline 
{\bf Graph} & $\mathbf{\eta}$ & $\mathbf{n}$ & $\mathbf{m}$ & \parbox[t]{1.1cm}{\centering \bf avg. degree}\\
\hline
$G_{100, 0, M}$ & 0.0 & 359 & 1302 & 3.627 \\
$G_{100, 0.01, M}$ & 0.01 & 287 & 954 & 3.324\\
$G_{100, 0.1, M}$ & 0.1 & 139 & 323 & 2.324\\
\hline
$G_{500, 0, M}$ & 0.0 & 111 & 1196 & 10.775 \\
$G_{500, 0.01, M}$ & 0.01 & 107 & 1068 & 9.981\\
$G_{500, 0.1, M}$ & 0.1 & 75 & 272 & 3.627\\
\hline
\end{tabular}
\end{minipage}\hfill
\begin{minipage}{0.48\textwidth}
\begin{tabular}{|l| r| r| r| r| r|}
\hline 
{\bf Graph} & $\mathbf{\eta}$ & $\mathbf{n}$ & $\mathbf{m}$ & \parbox[t]{1.1cm}{\centering \bf avg. degree} \\
\hline
$G_{100, 0, E}$ & 0.0 & 1187 & 5854 & 4.932\\
$G_{100, 0.01, E}$ & 0.01 & 1099 & 4986 & 4.537\\
$G_{100, 0.1, E}$ & 0.1 & 222 & 463 & 2.086\\
\hline
$G_{500, 0, E}$ & 0.0 & 142 & 2625 & 18.486 \\
$G_{500, 0.01, E}$ & 0.01 & 125 & 1810 & 14.480\\
$G_{500, 0.1, E}$ & 0.1 & 92 & 229 & 2.489\\
\hline
\end{tabular}
\end{minipage}
\caption{Properties of the graphs used for the analysis. A graph $G_{s, r, \eta}$ with $s\in\{100\; m,500\; m\}, r\in\{M,E\}, \eta=\{0, 0.01, 0.1\}$ was obtained with a grid of size $s$, with morning (M) or evening (E) rides, and pruning factor $\eta$.}\tablab{graphdata}
\end{table}

\subsection{Performance and quality} \seclab{expeval}
The experimental analysis focuses on the following questions: 
\begin{description}
\item[(Q1)] How close is the solution of the approximate algorithm to the optimal solution?
\item[(Q2)] How does the greedy algorithm scale with input size, seed size, and step number?
\item[(Q3)] How do the \sbs and \tbs models compare?
\item[(Q4)] When bikes should be rebalanced?
\end{description}

All experiments have been executed on an Intel Xeon Processor W-2245 3.9GHz with 128GB RAM and Ubuntu 9.3.0; code was in Python 3. Running times were averaged on 3 executions.

\begin{table*}
\centering
\tiny
 \begin{tabular}{|l|c|r|r|r|r|r|r|r|r|}
		\hline
		& &\multicolumn{2}{|c|}{\parbox[t]{1.8cm}{\centering \bf Brute force $k=2$}} &\multicolumn{2}{|c|}{\parbox[t]{1.8cm}{\centering \bf Greedy $k=2$}}&\multicolumn{2}{|c|}{\parbox[t]{1.8cm}{\centering \bf Brute force $k=4$\vspace{0.5em}}} &\multicolumn{2}{|c|}{\parbox[t]{1.8cm}{\centering 	\bf Greedy $k=4$}}\\
		\cline{3-10}
		{\bf Graph} & $\mathbf{n/m}$ & \parbox[t]{0.8cm}{\centering 	$\mathbf{\sigma}$} & {\bf \parbox[t]{0.8cm}{\centering time [s]}} & \parbox[t]{0.8cm}{\centering 	$\mathbf{\sigma}$} & {\bf \parbox[t]{0.8cm}{\centering time [s]}} & \parbox[t]{0.8cm}{\centering 	$\mathbf{\sigma}$} & {\bf \parbox[t]{0.8cm}{\centering time [s]}} & \parbox[t]{0.8cm}{\centering 	$\mathbf{\sigma}$} & {\bf \parbox[t]{0.8cm}{\centering time [s]}}\\
		\hline
		$G_{500,0.1,M}$ & 75/272 & 32.6 & 0.02 & 32.6 & 0.001 & 43.6 & 7.20 & 42.8 & 0.001\\
		$G_{500,0.01,M}$ & 107/1068 & 55.3 & 0.07 & 55.3 & 0.002 & 63.9 & 71.36 & 63.9 & 0.005\\
		$G_{500,0,M}$ & 111/1196 & 57.3 & 0.07 & 57.3 & 0.002 & 64.1 & 86.16 & 63.8 & 0.005\\
		$G_{100,0,M}$ & 359/1302 & 121.0 & 1.79 & 121.0 & 0.021 & 173.9 & 18514 & 123.0 & 0.040\\
		$G_{100,0,E}$ & 1187/5854 & 185.7 & 99.00 & 185.7 & 0.347 & * & * & 193.3 & 0.555\\
		\hline
	\end{tabular}
	
	\caption{Comparison between the brute force and greedy algorithms for the \sbs version, with seed size $k\in\{2, 4\}$, $\tau=1$, $L=100$. The symbol $*$ means that the instance has not be run due to excessive running time.}
	\tablab{brute_greedy_s}
\end{table*}
	
\begin{table*}
\centering
\tiny
\begin{tabular}{|l|c|r|r|r|r|r|r|r|r|}
		\hline
		& &\multicolumn{2}{|c|}{\parbox[t]{1.8cm}{\centering \bf Brute force $k=2$}} &\multicolumn{2}{|c|}{\parbox[t]{1.8cm}{\centering \bf Greedy $k=2$}}&\multicolumn{2}{|c|}{\parbox[t]{1.8cm}{\centering \bf Brute force $k=4$\vspace{0.5em}}} &\multicolumn{2}{|c|}{\parbox[t]{1.8cm}{\centering 	\bf Greedy $k=4$}}\\
		\cline{3-10}
		{\bf Graph} & $\mathbf{n/m}$ & \parbox[t]{0.8cm}{\centering 	$\mathbf{\sigma}$} & {\bf \parbox[t]{0.8cm}{\centering time [s]}} & \parbox[t]{0.8cm}{\centering 	$\mathbf{\sigma}$} & {\bf \parbox[t]{0.8cm}{\centering time [s]}} & \parbox[t]{0.8cm}{\centering 	$\mathbf{\sigma}$} & {\bf \parbox[t]{0.8cm}{\centering time [s]}} & \parbox[t]{0.8cm}{\centering 	$\mathbf{\sigma}$} & {\bf \parbox[t]{0.8cm}{\centering time [s]}}\\
		\hline
		$G_{500,0.1,M}$ & 75/272 & 11 & 0.04 & 11 & 0.002 
		& 20 & 15.49 & 20 & 0.003\\
		$G_{500,0.01,M}$ & 107/1068 & 31 & 0.22 & 30 & 0.008 
		& 35 & 196.93 & 35 & 0.015\\
		$G_{500,0,M}$ & 111/1196 & 31 & 0.24 & 30 & 0.008 
		& 36 & 234.79 & 35 & 0.016\\
		$G_{100,0,M}$ & 359/1302 & 51 & 8.14 & 51 & 0.089 
		& * & * & 57 & 0.176\\
		$G_{100,0,E}$ & 1187/5854 & 81 & 302.17 & 81 & 1.03 
		& * & * & 84 & 2.06\\
		\hline
	\end{tabular}	
	
	\caption{Comparison between the brute force and greedy algorithms for the \tbs version, with seed size $k\in\{2, 4\}$, $\tau=1$, $L=100$, $\gamma=1$. The symbol $*$ means that the instance has not be run due to excessive running time.}
	\tablab{brute_greedy_t}
\end{table*}

\paragraph*{Question Q1: exact vs approximate solutions}
\tabref{brute_greedy_s} shows the spread and running time for the brute force exact algorithm and the greedy approximation algorithm with the \sbs version: since the brute force has $\BO{n^k}$ time, which is exponential in seed size, we notice a quick increase of the running time even for small seed size.
The greedy algorithm is more performing than the brute force, although it is visible the $n^3$ dependence in the running time.
The quality of the approximation is quite high, outperforming the theoretical worst-case upper bound of $1-1/e$. 
Similar results hold for \tbs as provided in \tabref{brute_greedy_t}: the running times show a small increase with respect to the previous table for \sbs, mainly due to the conditional statements for checking if a load is smaller than the threshold $\gamma$. 
In all cases, the spreads are very close, although the seed provided by the brute force and the greedy algorithms do not completely overlap. (see \tabref{seedbrute}).

\begin{table}
\centering
\begin{tabular}{|l | r| r |} 
\hline
{\bf Graph} & {\bf Seed overlap under \sbs} & {\bf Seed overlap under \tbs} \\
\hline
$G_{500,0.1,M}$  & 25\% & 50\%*\\
$G_{500,0.01,M}$ &  100\%* & 75\%*\\
$G_{500,0,M}$    & 25\% & 25\%\\
\hline
\end{tabular}
\caption{Seed comparison between the brute force algorithm and the approximate algorithm for the \sbs and \tbs models for $k=4$ (other parameters: $\tau=1$, $L=100$, $\gamma=1$). The star $^*$ after a seed set means that the solution provided by the approximate algorithm has the same spread of the optimal one.}\tablab{seedbrute}
\end{table}

\paragraph*{Question Q2: scalability}
We now analyze the running time of the greedy approach for increasing values of input size, number of steps, and seed size. 
Since the greedy algorithms for \sbs and \tbs are almost equivalent, we only provide results for the first one.
\tabref{tau} shows the running time on the largest graph ($G_{100,0,E}$), three different seed sizes $k\in\{2,4,8\}$ and three different step numbers $\tau\in\{1,10,100\}$. 
For a given $k$, the times are almost equivalent: $\tau$ only affects the initial computation of $A^\tau$, where $A$ is the adjacency matrix of the graph;
since the powering requires time $\BO{n^3\log \tau}$, the logarithmic dependence on $\tau$ is negligible and hidden by the cost of the $k$ iterations of the greedy algorithm.
The previous table already shows that the running time has a linear dependency on the seed size. 
\figref{seedsize} expands this analysis by considering different graph sizes. 
Note that the curve of $G_{100,0,E}$ has been scaled by a factor of 10 to fit the plot space.
All curves show a  linear dependency in $k$.

\begin{figure}
\begin{minipage}{0.5\textwidth}
	\centering
 \begin{tabular}{|l|l|r|r|r|}
		\hline
		{\bf Graph} & $\mathbf{k}$ & \parbox[t]{1cm}{\bf \centering $\mathbf{\tau=1}$ [ms]}& \parbox[t]{1cm}{\bf \centering $\mathbf{\tau=10}$ [ms]} & \parbox[t]{1.2cm}{\bf \centering $\mathbf{\tau=100}$ [ms]}\\
		\hline
		$G_{100,0,E}$ & 2 & 334.2 & 352.4 & 384.4 \\
		$G_{100,0,E}$ & 4 & 629.8 & 651.0 & 657.4 \\
		$G_{100,0,E}$ & 8 & 1271.8 & 1277.2 & 1278.8 \\
		\hline
	\end{tabular}
	\captionof{table}{Running time using $G_{100,0,E}$ with different values of the seed set size $k$ and of the step number $\tau$.}
\tablab{tau}
\end{minipage}	
\hfill
\begin{minipage}{0.45\textwidth}	
\centering
\begin{tikzpicture}[scale=0.50]
\begin{axis}[
 xlabel={Seed size},
 ylabel={Time [ms]},
 xmin=0, xmax=32,
 ymin=0, ymax=600,
 xtick={2,4,8,16,32},
 ytick={100,200,300,400,500,600},
 legend pos=north west,
 ymajorgrids=true,
 grid style=dashed,
]
\addplot
 coordinates {
 (2,4.87)	(4,9.29)(8,17.85)(16,31.42)(32,57.23)
 };
\addplot
 coordinates {
 (2,14.92)(4,24.18)(8,46.06)(16,85.66)(32,160.17)
 };
\addplot
 coordinates {
 (2,21.48)(4,37.94)(8,71.60)(16,138.10)(32,262.21)
 };
\addplot
 coordinates {
 (2,21.48)(4,37.94)(8,71.60)(16,138.10)(32,262.21)
 }; 
\addplot
 coordinates {
 (2,35.86)(4,47.20)(8,95.03)(16,280.17)(32,556.85)
 }; 				 
 \legend{$G_{100, 0.1, M}$,$G_{100, 0.01, M}$,$G_{100, 0, M}$, $G_{100, 0, E}$ (time x10)} 
\end{axis}
\end{tikzpicture}
\captionof{figure}{Running time with respect to seed size $k$. The curve of $G_{100,0,M}$ has been scaled by a factor x10 for better fitting.
}
\figlab{seedsize}
\end{minipage}
\end{figure}

\paragraph*{Question Q3: \tbs vs \sbs}
In this experiment, we compare the two models. 
Intuitively, \tbs maximizes the number of nodes with a minimum number $\gamma$ of bikes: for instance, by setting $\gamma=1$, the algorithm maximizes the number of cells with at least one bike (in expectation).
On the other hand, \sbs aims at uniformly distributing bikes among nodes, even if this implies that some nodes have a low expected number of bikes (even $<1$).
This allows for more fair use of bikes since it increases the load in suburb areas, differently than \tbs that facilitates central (and more crowded) areas.
In the heatmaps in \figref{spreadarea}, we compare the two models on the $G_{500,0,M}$ (heatmaps (a) and (b)) and $G_{100,0,M}$ (heatmaps (d) and (e)) graphs. 
Each heatmap shows how bike distributes, after $\tau=2$ steps, by positioning 400 bikes in the $k=4$ seed set selected by the greedy algorithm.
We notice that \sbs covers a larger fraction of nodes in the map than \tbs.
The phenomenon is more evident in the 100m grid, where \sbs colors a larger number of cells in the east and south suburb areas than \tbs. 
although, the majority of the cells are reached by a small number of bikes, mostly less than one bike in expectation.

\begin{figure*}[t]
 \begin{subfigure}{0.30\textwidth}
 \centering
 \includegraphics[width=1\linewidth]{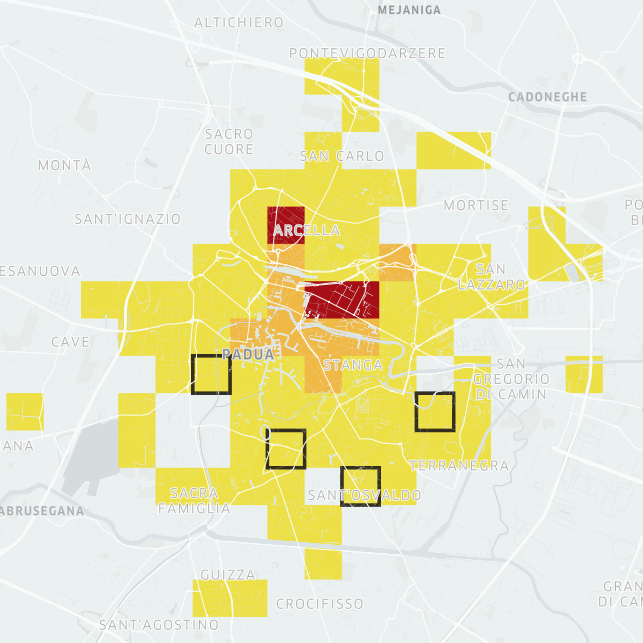}
 \caption{\tbs on 500 m, morning}
 \end{subfigure}\hfill
 \begin{subfigure}{0.30\textwidth}
 \centering
 \includegraphics[width=1\linewidth]{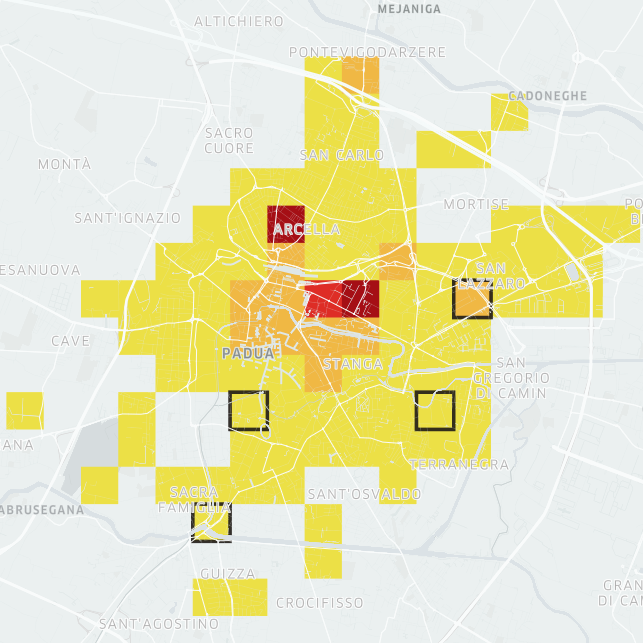}
 \caption{\sbs on 500 m, morning}
 \end{subfigure}\hfill
 \begin{subfigure}{0.30\textwidth}
 \centering
 \includegraphics[width=1\linewidth]{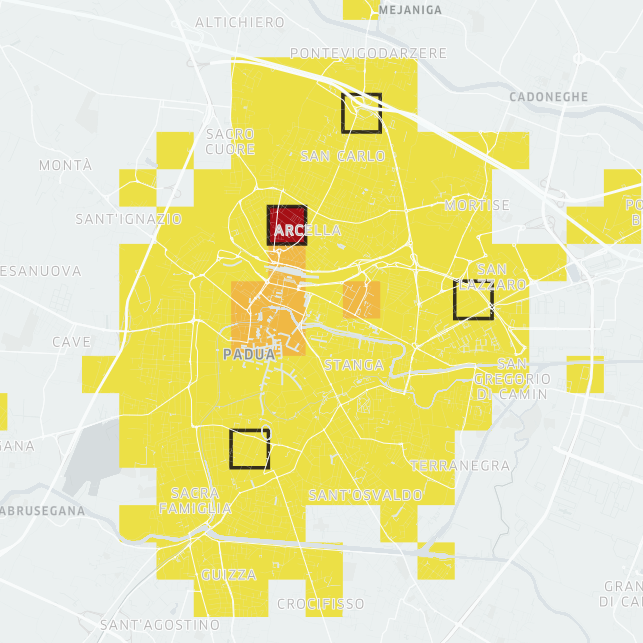}
 \caption{\sbs on 500 m, evening}
 \end{subfigure}
 \hfill
 \\

 \begin{subfigure}{0.30\textwidth}
 \centering
 \includegraphics[width=1\linewidth]{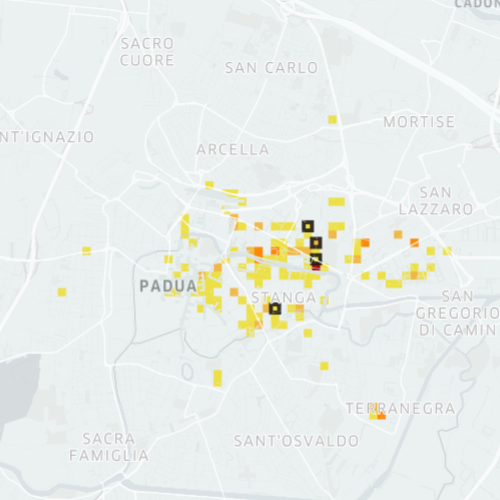}
 \caption{\tbs on 100 m, morning}
 \end{subfigure}\hfill
 \begin{subfigure}{0.30\textwidth}
 \centering
 \includegraphics[width=1\linewidth]{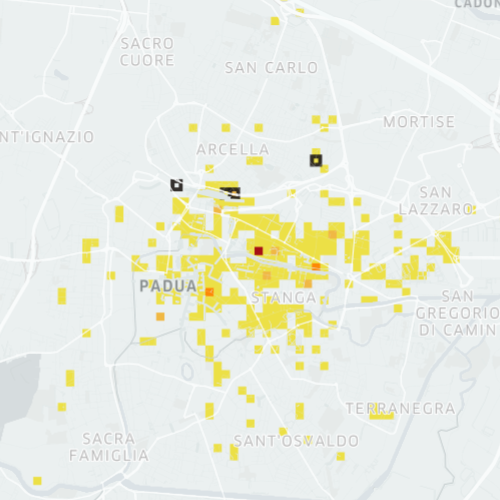}
 \caption{\sbs on 100 m, morning}
 \end{subfigure}\hfill
 \begin{subfigure}{0.30\textwidth}
 \centering
 \includegraphics[width=1\linewidth]{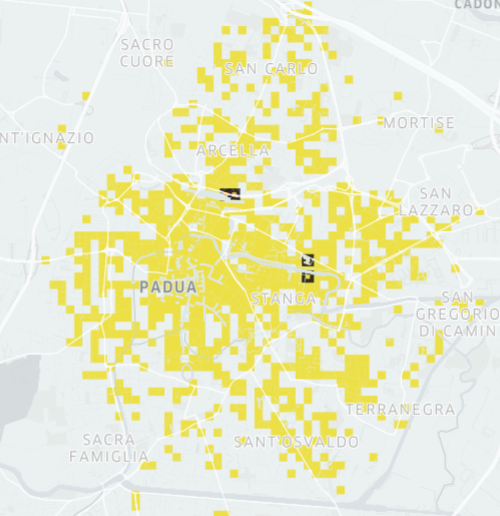}
 \caption{\sbs on 100 m, evening}
 \end{subfigure}
 \caption{Bike diffusion after $\tau=2$ steps by positioning $L=100$ bikes in each node of the seed set of size $k=4$ selected by the greedy algorithm. For \tbs, we set $\gamma=1$. The cells in the seed set are marked with black contours. Color scale for 500m plots: $(0.0, 2.4]$ (yellow),$(2.4,4.8],(4.8,7.2],(7.2, 9.6], (9.6,12]$ (dark red); Color scale for 100m plots: $(0.0, 1.4]$ (yellow), $(1.4,2.8],(2.8,4.2],(4.2, 5.6], (5.6,7]$ (dark red). (The underlying street maps were provided by Kepler.gl.)} 
\figlab{spreadarea}
\end{figure*}

\paragraph*{Question Q4: When rebalancing?}
Consider the graphs $G_{100,0,M}$ and $G_{100,0,E}$: both graphs use the 100m grid and no edge pruning. 
Graph $G_{100,0,E}$  has a 
larger number of edges and nodes than $G_{100,0,M}$ (see \tabref{graphdata}), highlighting a change in the use of the FFBSS service from morning and evening: in the morning, rides are mostly directed towards workplaces, the train station, and university departments which are mainly located in the city center and on the east side; in the afternoon, there are much more activities (e.g., having a \emph{spritz} with friends, going to the gym, shopping) and hence the graph covers a wider area of Padova.
Bikes can be more spread in the city if we use the more vibrant afternoon for rebalancing bikes.
This is confirmed by the simulation of \sbs using $G_{100,0,M}$ and $G_{100,0,E}$ in the heatmaps (e) and (f) of \figref{spreadarea}: by using the evening graph, the bikes significantly spread covering also the north and west parts of the city. Similar results hold for $G_{500,0,M}$ and $G_{500,0,E}$ in heatmaps (b) and (c).

\section{Conclusion}
In this work, we have proposed a graph approach to spread bikes in a free-floating bike system to cover a large number of zones of the service area; the idea is to select a set of zones where to position bikes and let the mobility flow spread them around the city.
The current model assumes that, initially, only seed nodes have bikes while the other nodes are empty. This assumption can be removed by allowing any node to initially have some bikes: it can be shown that submodularity and monotonicity still hold for \sbs, and thus the greedy algorithm provides a $1-1/e$ approximate solution also in this case.
Another extension to investigate is to allow a different distribution of bikes among seed nodes, to balance skewness in bike distribution due to hotspots.
Finally, an important research direction is to analyze the behavior of the model on a real free-floating bike system and the differences between the theoretical findings and the actual distribution.

\bibliography{biblio}
\end{document}